\documentclass[a4paper]{lipics-v2021}
\usepackage{graphicx}
\usepackage{tikz}
\usepackage{amsfonts}
\usepackage{float}
\usepackage{xcolor}
\usepackage{amsmath}
\usepackage{mathabx}
\usepackage{hyperref}

\usepackage{algpseudocode}
\usepackage[ruled]{algorithm}
\algrenewcommand\algorithmicrequire{\textbf{Input:}}
\algrenewcommand\algorithmicensure{\textbf{Output:}}
\algtext*{EndWhile}
\algtext*{EndIf}
\algtext*{EndFor}
\algtext*{EndFunction}
\newcommand{\ovrt}{\ominus}
\newcommand{\ohrz}{\mathbin{\rotatebox[origin=c]{90}{$\ovrt$}}}
\newcommand{\no}[1]{}
\newcommand{\vir}[1]{``#1''}
\newcommand{\zero}{\texttt{0}}
\newcommand{\one}{\texttt{1}}

\newcommand{\dd}{\mathinner{.\,.}}
\newcommand{\Mmn}{M_{m\times n}}
\newcommand{\Mnn}{M_{n\times n}}
\newcommand{\mxn}{[1\dd m] \times [1\dd n]}
\newcommand{\gexp}{\mathtt{exp}}
\newcommand{\rowlin}{\mathtt{rlin}}
\newcommand{\rhs}{\mathtt{rhs}}
\newcommand{\size}{\mathtt{size}}
\newcommand{\map}{\mathtt{map}}

\newcommand{\deltaCM}{\delta_\square}
\usetikzlibrary{positioning, matrix, arrows}
\usepackage{ifthen}

\definecolor{DarkGreen}{RGB}{40,160,20}

\newcommand{\GRremark}[1]{\marginpar{\tiny \flushleft{[GR]~#1}}}

\title{Exploring Repetitiveness Measures for Two-Dimensional Strings}

\author{Giuseppe Romana}{Dipartimento di Matematica e Informatica, Università di Palermo, Italy}{giuseppe.romana01@unipa.it}{https://orcid.org/0000-0002-3489-0684}{Partly supported by the MUR PRIN Project \vir{PINC, Pangenome INformatiCs: from Theory to Applications} (Grant No.\ 2022YRB97K)}

\author{Marinella Sciortino}{Dipartimento di Matematica e Informatica, Università di Palermo, Italy}{marinella.sciortino@unipa.it}{https://orcid.org/0000-0001-6928-0168}{Partly supported by the MUR PRIN Project \vir{PINC, Pangenome INformatiCs: from Theory to Applications} (Grant No.\ 2022YRB97K)}

\author{Cristian Urbina}{Department of Computer Science, University of Chile, Chile \and Centre for Biotechnology and Bioengineering (CeBiB), Chile}{crurbina@dcc.uchile.cl}{https://orcid.org/0000-0001-8979-9055}{Partly supported by ANID-Subdirección de Capital Humano/Doctorado Nacional/ 2021-21210580}

\authorrunning{G. Romana, M. Sciortino, and C. Urbina} 

\Copyright{Giuseppe Romana, Marinella Sciortino, and Cristian Urbina} 

\ccsdesc[100]{Mathematics of computing~Combinatorics on words}
\ccsdesc[100]{Theory of computation~Data compression}

\keywords{Two--dimensional strings, Repetitiveness measures, Text compression} 

\category{} 

\relatedversion{} 

\acknowledgements{This research was carried out during the visit of Cristian Urbina to the University of Palermo, Italy.}

\nolinenumbers 

\EventEditors{John Q. Open and Joan R. Access}
\EventNoEds{2}
\EventLongTitle{42nd Conference on Very Important Topics (CVIT 2016)}
\EventShortTitle{CVIT 2016}
\EventAcronym{CVIT}
\EventYear{2016}
\EventDate{December 24--27, 2016}
\EventLocation{Little Whinging, United Kingdom}
\EventLogo{}
\SeriesVolume{42}
\ArticleNo{23}

\begin{document}

\maketitle
\begin{abstract}
Detecting and measuring repetitiveness of strings is a problem that has been extensively studied in data compression and text indexing.
However, when the data are structured in a non-linear way, like in the context of two-dimensional strings, inherent redundancy offers a rich source for compression, yet systematic studies on repetitiveness measures are still lacking.
In the paper we introduce extensions of repetitiveness measures to general two-dimensional strings.
In particular, we propose a new extension of the measures $\delta$ and $\gamma$, diverging from previous square-based definitions proposed in [Carfagna and Manzini, SPIRE 2023].
We further consider generalizations of macro schemes and straight line programs for the 2D setting and show that, in contrast to what happens on strings, 2D macro schemes and 2D SLPs can be both asymptotically smaller than $\delta$ and $\gamma$. 
The results of the paper can be easily extended to $d$-dimensional strings with $d > 2$.
\end{abstract}

\section{Introduction}

In the latest decades, the amount of data generated in the world has become massive. Nevertheless, in many fields, most of this data is highly repetitive. Repetitiveness has been showed to be an exploitable source of compressibility. In fact, compressors exploiting repetitiveness can perform much better  on repetitive datasets like genome collections, compared to classic compressors approaching empirical entropy. 


Two-dimensional data, ranging from images to matrices, often contains inherent redundancy, wherein identical or similar substructures recur throughout the dataset. This great source of redundancy can be exploited for compression. Very recently, Brisaboa et al. introduced the 2D Block Trees to compress images, graphs, and maps~\cite{BrisaboaGGN24}. 
On the theoretical side, while in the one-dimensional case much attention has been given to the study and analysis of measures of repetitiveness to assess the performance of compressed indexing data structures~\cite{NavarroSurvey}, in the two-dimensional context, since the data are more complex, there is still no systematic study of measures that can effectively capture repetitiveness.
It must also be said that in the one-dimensional case, an important role is played by the $\delta$ measure, which computes the maximum number of substrings of the same length that occur in a text, and $\gamma$ measure, which represents the smallest number of positions (string attractors) in the text at which all substrings can be found. 
These measures, although unreachable or unknown to be reachable, lower-bound all other repetitiveness measures based on \emph{copy-paste} mechanisms.
Furthermore, the optimal space to represent a text can be expressed as a function of $\delta$~\cite{Delta}.
In the two-dimensional case, an important approach in this direction has been made by Carfagna and Manzini in~\cite{CarfagnaManzini2023}, where the repetitiveness measures $\delta$ and $\gamma$ are extended to square two-dimensional strings, by exploring square substructures within the data. They have shown that some properties that hold for one-dimensional strings are still preserved in the two-dimensional case, and the space used by a two-dimensional block tree has been bounded in terms of the extension of $\delta$.



In this paper, we propose extensions of repetitiveness measures to generic two-dimensional strings. First of all, we introduce the measures $\delta_{2D}$ and $\gamma_{2D}$, which differently from the measures defined in~\cite{CarfagnaManzini2023}, use rectangular substrings, instead of square, in their definition. We show that our measures, while retaining many of the properties valid in one-dimensional case, can exhibit a significant gap when compared to those defined using the strategy proposed in~\cite{CarfagnaManzini2023}, even if applied to one-dimensional strings.

Furthermore, we generalize straight-line programs (SLPs) and run-length straight-line programs (RLSLPs). In particular, we introduce a new repetitiveness measure $g_{2D}$ based on SLP and we show that, although it is NP-hard to compute, it is possible to access to an arbitrary position of the 2D strings in $O(g_{{rl}_{2D}})$ space and $O(h)$ time, where $h$ is the length of the parse tree of the 2D-SLP. Note that the time to access any cell can become logarithmic in the size of the 2D string when a balancing procedure is applied~\cite{GJL2021}. An analogous result can be obtained when 2D-RLSLP and the correspondent repetitiveness measure $g_{{rl}_{2D}}$ are considered. 
Finally, we introduce macro schemes for 2D strings that generalize bidirectional macro schemes. We show that the mutual relationship among $g_{2D}$, $g_{{rl}_{2D}}$ and the size $b_{2D}$ of the smallest valid 2D macro scheme are the same as for one-dimensional strings. 

However, we show that some relevant relationship between $\delta$, $\gamma$ and the other repetitiveness measures are lost when they are extended to the 2D setting. For instance, it is well known that for 1D strings the relationship $\delta \leq \gamma \leq b\leq g_{rl} \leq g$ holds. On the other hand, in the 2D setting, it can happen that $\delta_{2D} = \Omega(g_{2D}\sqrt[4]{N}/\log N)$ for some 2D string families, where $N$ is the size of the 2D strings. 
This bound is still valid also when the strategy proposed in~\cite{CarfagnaManzini2023} is used.

For space reasons, in this paper we focus only on 2D strings. Note that the repetitiveness measure defined and studied in this paper can be naturally extended when $d$-dimensional strings with $d > 2$ are considered. In this context, all the results provided in the paper are still valid. More extensive and detailed results on higher dimensional strings are deferred to the full paper.
%

\section{Preliminaries}

We intend to extend repetitiveness measures from 1D strings to 2D strings, so in this section we recall such notions for the 1D setting. We only describe the repetitiveness measures we are focusing on. For an in-depth survey, the reader can see~\cite{NavarroSurvey}.  We assume the RAM model of computation with words of $\Theta(\log n)$ bits. Hence, $O(\mu)$ space also means $O(\mu \log n)$ bits. 

\paragraph*{Strings}

Let $\Sigma = \{a_1 \dd a_\sigma\}$ be a finite ordered set of \emph{symbols}, which we call an \emph{alphabet}.
A \emph{string}\footnote{We use the concepts of \emph{string}, \emph{text}, and \emph{word},  interchangeably.} $S$ is a finite sequence of symbols from the alphabet $\Sigma$, and we denote by $S[i]\in\Sigma$ the $i$th symbol in $S$.
The \emph{length} of a string $S$, denoted by $|S|$, is the number of symbols contained in $S$.
The \emph{empty string}, that is the only string of length 0, is denoted by $\varepsilon$.
We further denote by $S[i\dd j]$ the string in $S$ starting at position $1\leq i \leq |S|$ and ending at position $1\leq j \leq |S|$, that is $S[i\dd j] = S[i] S[i+1] \cdots S[j]$. When $i>j$, we assume that $S[i\dd j] = \varepsilon$.
Given two strings $S[1\dd n]$ and $T[1 \dd m]$, the \emph{concatenation} of $S$ and $T$ is the string $S\cdot T = S[1]S[2]\cdots S[n] \cdot T[1] T[2]\cdots T[m]$.
When the context is clear, the operator for the concatenation $\cdot$ is omitted.
The set of all finite strings over $\Sigma$ is denoted $\Sigma^*$.
A string $F$ is a \emph{factor} (or substring) of a string $S[1\dd n]$ if there exist two integer values $1\leq i,j \leq n$ such that $F = S[i\dd j]$. In this case, we say that the $F$ has an \emph{occurrence} at position $i$ in $S$.
The \emph{substring complexity} function of a string $S$, denoted by
$P_S$, counts for each integer $k>0$ the number of distinct factors of length $k$ in $S$. 

\paragraph*{Measure $\delta$}
The measure $\delta(T[1 \dd n])$ is defined as $\max_{k \in [1..n]}(p_T(k)/k)$~\cite{RRR2013,Delta}. The measure $\delta$ has many desirable properties: i) it lower-bounds all the measures in the next paragraphs; ii) it is computable in linear time; iii) it is resistant to many string operations; iv) though $\delta$ is unreachable~\cite{Delta}, there exist algorithms producing representations of strings using $O(\delta \log (n /\delta))$ space, with strong indexing functionalities~\cite{Delta, KNO2023}. While there exist reachable repetitiveness measures that can be smaller than $\delta$ in some string families~\cite{NUspire21.1,NUcpm23.2,NUlatin24.2}, $\delta$ is still considered the gold standard when the only source of repetitiveness to be exploited are explicit copies.

\paragraph*{String Attractors and Measure $\gamma$}

A \emph{string attractor}~\cite{KP18} for a string $T[1\dd n]$ is a set of positions $\Gamma \subseteq [1\dd n]$ satisfying that for any substring $T[i\dd j]$, there exists a copy $T[i'\dd j']$ (i.e, $T[i\dd j] = T[i'\dd j']$) such that $i' \le k \le j'$ for some $k \in \Gamma$. Having a small string attractor implies that a few different positions in $T$ concentrate all the distinct substrings. The size of the smallest string attractor for a string is denoted $\gamma(T)$ and is considered a measure of repetitiveness. It is unknown to be reachable, and also NP-hard to compute~\cite{KP18}, though.

\paragraph*{(Run-Length) Straight-Line Programs and Measures $g$ and $g_{rl}$} A \emph{straight-line program} (SLP) is a context-free grammar $G = (V,\Sigma, R, S)$ (the elements of the tuple are the \emph{variables}, \emph{terminal symbols}, \emph{rules}, and the \emph{starting variable}, respectively) satisfying the following constraints: i) for each variable, there is only one rule $A \rightarrow \rhs(A)$; ii) all the rules have either the form $A \rightarrow a$ or the form $A \rightarrow BC$, where $a$ is a terminal symbol and $B, C$ are variables; iii) there is a total order between variables such that if $B$ appear in $\rhs(A)$, then $B > A$. These restrictions ensure that any variable $A$ can derive a unique string of terminals, namely $\gexp(A)$, which is defined as $\gexp(A) = a$ if $A \rightarrow a$, or $\gexp(A) = \gexp(B)\gexp(C)$ if $A \rightarrow BC$. Hence, each SLP generates a unique string $T[1 \dd n] = \gexp(G)=\gexp(S)$. The \emph{parse tree} of an SLP is an ordinal labeled tree where $S$ is the root, and the children of a variable $A$ are the variables in $\rhs(A)$ (possibly repeated), in left-to-right order. The \emph{height} of an SLP is the length of a maximal path from root to leaf in its parse tree. The \emph{grammar tree} is obtained by pruning all the subtrees with root $A$ (i.e., making $A$ a leaf node) in the parse tree of $G$, except one. If it is not specified, we assume that this exception is the leftmost occurrence of the variable. 
The size $\size(G)$ of an SLP $G$ can be defined as the number of its variables~\cite{Rytter2003}. Note that such a value is proportional to the space needed to represent the SLP. Given a text $T$, the measure of repetitiveness $g(T)$ is defined as the size of the smallest SLP generating $T$. Though computing $g$ is NP-hard~\cite{Charikar2005}, there exist algorithms producing $log$-approximations of the smallest grammar~\cite{Rytter2003,Jez2015}, and also heuristics like RePair~\cite{RePair}, which produce small grammars in practice. SLPs are a popular compression scheme because many relevant queries on the text, like direct access, locate, and count can be answered efficiently in compressed form. To provide these functionalities and more, usually having a \emph{balanced} SLP, i.e., an SLP with $O(\log n)$ height, comes in handy. It has been proven that any SLP can be balanced without asymptotically increasing its size~\cite{GJL2021}. It also holds that $g$ is upper-bounded by $O(\gamma\log^2(n/\gamma))$~\cite{KP18}.


A \emph{run-length straight-line program} (RLSLP) is an SLP that supports rules of the form $A \rightarrow B^k$, whose expansion is defined as $\gexp(A) = \gexp(B)^k$. This extension allows RLSLP, for instance, to represent strings in the family $\{\zero^n\,|\, n \geq 1\}$ using $O(1)$ space, which is asymptotically smaller than the $\Omega(\log n)$ space a regular SLP would need. There exist algorithms constructing small RLSLPs and providing strong functionality~\cite{Delta, KNO2023}. It also has been proven that RLSLPs can be balanced without asymptotically increasing their size~\cite{NOU2022}. It also holds that $g_{rl} = O(\delta \log (n/\delta))$~\cite{Delta}.

\paragraph*{Bidirectional Macro Schemes and Measure $b$ }
A \emph{bidirectional macro scheme}~\cite{SS82} is a factorization of a text $S[1\dd n]$ into substrings called \emph{phrases} $X_1 \dd X_r$ where each phrase of length greater than 1 has a source in $S$ from where it is copied starting at a different position. A phrase of length 1 is called an \emph{explicit symbol}. For a macro scheme to be  \emph{valid} or \emph{decodifiable}, there must exist a function $\map: [1 \dd n] \cup \{\bot\} \rightarrow [1 \dd n] \cup \{\bot\}$ satisfying that: i) $\map(\bot) = \bot$, and if $S[i]$ is a explicit symbol, then $\map(i) = \bot$; ii) for each non-explicit phrase  $S[i\dd j]$, it must hold that $\map(i+t) = \map(i) + t$ for $t \in [0\dd j-i]$; iii) for each $i \in [1\dd n]$ there exists $k > 0$ such that $\map^k(i) = \bot$. If this function exists, we can simulate it by storing the explicit symbols, and for each non-explicit phrase, its starting position in $S$ and its length. This allows us to recover the original text using $\Theta(r)$ space. The reachable repetitiveness measure $b$ is defined as the size of the smallest valid bidirectional macro scheme for $S$.

The measure $b$ is lower-bounded by $\gamma$, and sometimes it can be asymptotically greater~\cite{BFIKMN21}. It is easy to show via the grammar tree of an RLSLP that $b = O(g_{rl})$, and there exist string families where $b = o(g_{rl})$~\cite{NOP20}. On the negative side, it has been proven that computing $b$ is NP-hard~\cite{Gallant1982}. Moreover, no $O(\log^c n)$ time algorithm (for any constant $c$) for retrieving a symbol $T[i]$ has been found.

\section{2D Strings and Measure \texorpdfstring{$\delta$}{Delta}}

In this section, we define a repetitiveness measure for 2D strings with the goal of extending the $\delta$ measure to the two-dimensional context. Let us first give the definition of 2D strings and substring complexity in the two-dimensional case.

Let $\Sigma$ be an alphabet. A \emph{2D string} $M_{m \times n}$ 
is a ($m \times n$)-matrix with $m$ \emph{rows} and $n$ \emph{columns} such that each element $M[i][j]$ belongs to $\Sigma$. The \emph{}{size} of $\Mmn$ is $N = mn$. Note that a position in $\Mmn$ consists now in a pair $(i_1,i_2)$, with $1\leq i_1 \leq m$ and $1\leq i_2 \leq n$. 
Throughout the paper, we assume that for each 2D string $\Mmn$ it holds that $m,n\geq 1$.
We denote by $\Sigma^{m\times n}$ the set of all matrices with $m$ rows and $n$ columns over $\Sigma$.
The concatenation between two matrices is a partial operation that can be performed horizontally ($\ohrz$) or vertically ($\ovrt$), with the constraint that the number of rows or columns coincide respectively. Such operations have been described in~\cite{GiammarresiR97} where concepts and techniques of formal languages have been generalized to two dimensions.
We denote by $\Mmn[i_1 \dd j_1][i_2 \dd j_2]$ the rectangular submatrix starting at position $(i_1, i_2)$ and ending at position $(j_1, j_2)$, and we say that a matrix $F$ is a \emph{2D factor} of $\Mmn$ if there exist two positions $(i_1, i_2)$ and $(j_1, j_2)$ such that $F = \Mmn[i_1 \dd j_1][i_2 \dd j_2]$.
Analogously to strings, given a 2D string $\Mmn$ the \emph{2D substring complexity} function $P_M$ counts for each pair of positive integers $(k_1,k_2)$ the number of distinct ($k_1\times k_2$)-factors in $\Mmn$.






\begin{definition} \label{def:2ddelta}Let $\Mmn$ be a 2D string and $P_M$ be the 2D substring complexity of $\Mmn$. Then, $\delta_{2D}(\Mmn) = \max\{P_M(k_1, k_2)/k_1k_2, 1\leq k_1\leq m, 1\leq k_2\leq n\}$.
\end{definition}

Note that, for each string $S\in\Sigma^n$, there exists a corresponding 2D string $M_{1\times n}$ such that $M_{1\times n}[1][1\dd n] = S[1\dd n]$, and one can observe that $\delta_{2D}(M_{1\times n})=\delta(S)$. Since on one-dimensional matrices (when either $m=1$ or $n=1$) the two measures coincide, from now on we will use the notation $\delta$ to denote $\delta_{2D}$.



Recently, in~\cite{CarfagnaManzini2023} Carfagna and Manzini introduced an analogous extension of $\delta$, here denoted by $\delta_\square$, limited to square 2D strings and using only square factors for substring complexity. In their definition, for each possible $k$, $\delta_\square$ counts the number of distinct square matrices of size $k \times k$ in a square 2D string $\Mnn$ divided by $k^2$ and take the maximum. Below we report the definition of such a measure, applied to a generic two-dimensional string.


\begin{definition}
    Let $\Mmn$ be a 2D string and $P_M$ be the 2D substring complexity of $\Mmn$. Then, $\delta_\square(\Mmn) = \max\{P_M(k, k)/k^2, 1\leq k\leq \min\{m,n\}\}$.
\end{definition}

From the definitions of $\delta_\square$ and $\delta$, the following Lemma easily follows.

\begin{lemma}
For every 2D string $\Mmn$ it holds that $\delta(\Mmn) \geq \delta_\square(\Mmn)$.
\end{lemma}

However, the two measures $\delta$ and $\delta_\square$, although may seem similar, can have quite different behavior for several families of 2D strings. 
In fact, considering square factors instead of rectangular ones may result in the two measures having very different values.
In particular, Example~\ref{ex:deltasquare_on_strings} shows how different the two measures can be when applied to non-square 2D strings, e.g. one-dimensional strings. Example~\ref{ex:Carfagna_Manzini_Lemma4} shows that there exist families of square 2D strings for which $\delta_\square=o(\delta)$.

\begin{example}
\label{ex:deltasquare_on_strings}
    Given an alphabet $\Sigma$, let us consider a string $S\in\Sigma^n$, and the matrix $M_{1\times n}\in\Sigma^{1\times n}$ such that $M_{1 \times n}[1][1\dd n] = S[1 \dd n]$. One can observe that the only squares that occur in $M_{1\times n}$ are factors of size $1\times 1$, i.e. $\delta_\square(M_{m\times 1}) = P_M(1,1)/1^2 \leq |\Sigma|$. On the other hand, $\delta(M_{1\times n}) = \delta(S)$.
\end{example}

\begin{example}\label{ex:Carfagna_Manzini_Lemma4}
Let $M_{n\times n}$ be the square 2D string defined by Carfagna and Manzini in~\cite[Lemma 4]{CarfagnaManzini2023}. The first row of $M_{n\times n}$ is a string $S$ composed by $\sqrt{n}/2$ blocks, each one of size $2\sqrt{n}$ (assume $n$ is a perfect square). The $i$-th block $B_i$ for $i \in [1\dd \sqrt{n}/2]$ is the string $\one^i\zero^{(2\sqrt{n}-i)}$. Then, $S = B_1B_2 \dd B_{\sqrt{n}/2}$. The remaining rows of $M_{n \times n}$ are all $\texttt{\#}^n$. 
Carfagna and Manzini proved that $\deltaCM = O(1)$ in this string family~\cite[Lemma 4]{CarfagnaManzini2023}. 
On the other hand, notice that for $i \in [2\dd\sqrt{n}/2]$ and $j \in [0 \dd \sqrt{n}-i]$, the strings $\zero^j\one^i\zero^{\sqrt{n}-j-i}$ are all different substrings of length $\sqrt{n}$ of $S$. In total there are $\Omega(n)$ of them. Hence, $\delta(M_{n \times n}) = \Omega(\sqrt{n})$.
\end{example}

We conclude this section by observing that approaches based on linearizations of 2D strings may have a very different behaviour with respect to 
the 2D-extension of $\delta$, as the following examples show.

Example~\ref{ex:delta_row_linearization} considers the linearization $\rowlin$ that maps a matrix $\Mmn$ to the string obtained by concatenating all the rows of the matrix, i.e. $M[1][1..n]\cdot M[2][1..n]\cdots M[m][1..n]$.

\begin{example}\label{ex:delta_row_linearization}
    Let us consider a matrix $M_{n\times n}$, obtained by appending to the identity matrix $I_{n-1}$, a row of $\zero$'s at the bottom, and then a column of $\one$'s at the right, in that order. For each $k_1$ and $k_2$, the value $P_M(k_1, k_2)$ is at most $3(k_1+ k_2 )$. We can see this by considering three cases: the submatrices that do not intersect the last row nor the last column of $M$, the submatrices aligned with the row of $\zero$'s at the bottom, and the submatrices aligned with the column of $\one$'s at the right. In each case, the distinct submatrices are associated to where the diagonal of $\one$'s intersects a submatrix (if it does so). This can happen in at most $k_1 + k_2$ different ways. As $3(k_1+k_2)/k_1k_2 \le 6$, we obtain $\delta(\Mnn)= O(1)$. On the other hand, for each $k \in [1\dd n]$ and $i \in [0\dd n-k-2]$, each factor $\zero^i\one\zero^k\one\zero^{n-k-i-2}$ appears in $\rowlin(\Mnn)$. There are $n-k-1$ of these factors for each $k$. Summing over all $k$, we obtain $P_M(n) \ge (n^2-3n)/2$, hence $P_M(n)/n = \Omega(n)$. Thus, $\delta(\rowlin(\Mnn))= \Omega(n)$.
\end{example}

Indeed, for each 2D string $\Mmn\in\Sigma^{m\times n}$ there exists a linearization which is highly compressible.
For instance, we can visit in order all the occurrences of $a_1\in\Sigma$, followed by all the occurrence of $a_2\in\Sigma$, and so on until we get the string which consists in $|\Sigma|$ equal-letter runs.
Nonetheless, a universal linearization which reduces the size of any 1D measure does not seem to exist.
Similar considerations can be made when other repetitiveness measures are used. That is, using measures defined for one-dimensional strings to assess the repetitiveness of two-dimensional strings through linearization strategies may differ from the 2D definitions of such measures. 
More detailed examples are provided in Appendix \ref{sec:linearization}.

\section{2D String Attractors}

In this section, we define an analogous to string attractors for the two-dimensional setting. 

\begin{definition}\label{def:2dstrinattractor}A \emph{2D string attractor} for a 2D string $M[1\dd m][1 \dd n]$ is a set $\Gamma \subseteq \mxn$ of positions in $\Mmn$, satisfying that any substring $M[i\dd j][k \dd l]$ has an occurrence $M[i'\dd j'][k' \dd l']$ such that $i' \le x \le j'$ and $k' \le y \le l'$ for some $(x, y) \in \Gamma$. The size of the smallest 2D string attractor for $\Mmn$ is denoted by $\gamma_{2D}(\Mmn)$.
\end{definition}

It is easy to see that, analogously to $\delta$, Definition~\ref{def:2dstrinattractor} reduces to definition of $\gamma$ on 1D strings when either $m = 1$ or $n=1$.
Hence, also in case of string attractors, we will use $\gamma$ to denote $\gamma_{2D}$.
By using this simple observation, we can conclude that the measure $\gamma$ is not monotone. Moreover, the following proposition follows.

\begin{proposition}Let $\Mmn$ be a 2D string. Computing $\gamma(\Mmn)$ is NP-hard.
\end{proposition}

By using similar arguments to the one-dimensional case, the following proposition can be proved.

\begin{proposition}
For every 2D string $\Mmn$, it holds that that $\delta(M)\leq \gamma(M)$.   
\end{proposition}
\begin{proof}
  It follows from the fact that $P_M(k_1,k_2)\leq k_1 k_2 \gamma(M)$, for any 2D string $M$.  
\end{proof}
The following proposition, similarly to what happens for 1D strings, shows that $\delta$ can be $o(\gamma)$. The proposition also shoes that in the 2D context the gap between $\delta$ and $\gamma$ can be larger that the one-dimensional case, where it is logarithmic~\cite{Delta}.

\begin{proposition}\label{prop:gamma_odelta} There exist a 2D string family where $\gamma(\Mmn) = \Omega(\delta(\Mmn) \max(m,n))$.
\end{proposition}

\begin{proof}Consider the 2D string $I_{m}$ defined as the identity matrix of order $m$. 
One 2D string attractor for this string is the set $\Gamma = \{(1,1) \dd (m,m)\} \cup \{(1, m)\}$. This is because a substring either contains $\one$'s, all of them lying at an attractor position, or consists of only $\zero$'s, and it has an occurrence aligned with the $\zero$ at the upper-right corner of the string. Observe that any row and column has to contain an attractor position because all rows and columns are distinct. The only way to satisfy this constraint using $m$ positions is marking all the $\one$'s, but then it is still necessary to mark a $\zero$ for the substrings of $I_{m}$ that do not contain $\one$'s. Thus, the string attractor above has minimum size and $\gamma(I_{m}) = m+1$.

On the other hand, there exist at most $k_1+k_2$ distinct substrings of size $k_1 \times k_2$ in $I_m$; $k_1+k_2 -1$ correspond to substrings where the diagonal of $I_m$ touches ones of its left or upper borders; the last one is the string of only $\zero$'s.  Hence, $\delta = O(1)$ in this string family.
\end{proof}

Analogously to what we have seen in the previous section, Carfagna and Manzini~\cite{CarfagnaManzini2023} introduced a definition of string attractors for square 2D strings in which they consider only square factors. We can define such a measure, denoted by $\gamma_\square$, also for generic 2D strings, by simply considering square substrings in Definition~\ref{def:2dstrinattractor}.
From the definitions of $\gamma$ and $\gamma_\square$, the following relationship between these measures hold. 

\begin{lemma}
  For every 2D string $\Mmn$ it holds that $\gamma(\Mmn) \geq \gamma_\square(\Mmn)$.  
\end{lemma}

As the following example shows, these two measures can be asymptotically different.

\begin{example}
    Consider once again the $m$th order identity matrix $I_m$. 
    Observe that, for each $k\leq m$, a $k\times k$ square factor of $I_m$ either consists of i) all 0's, or ii) all 0's except only one diagonal composed by 1's.
    Hence, all 2D square factors falling in i) have an occurrence that includes position $(m,1)$ (i.e. the bottom left corner), while all those falling in ii) have an occurrence that includes the position $(\lfloor m/2 \rfloor, \lfloor m/2 \rfloor)$ (i.e. the 1 at the center).
    It follows that $\gamma_\square(I_m) = 2 \in O(1)$, while as already showed in the proof of Proposition~\ref{prop:gamma_odelta} one has that $\gamma(I_m) = \Theta(m)$.
\end{example}

Nonetheless, we want to remark once again that considering only square factors could be a weak definition for the general case, since also in this case for each $\Mmn\in \Sigma^{m\times n}$ with either $m=1$ or $n=1$, it holds that $\gamma_\square(M)\leq |\Sigma|$.

\section{(Run-Length) Straight-Line Programs for 2D Strings}

We define a generalization of SLPs suitable for the two-dimensional space and introduce a new measure of repetitiveness for 2D texts based on it. Although this measure is NP-hard to compute, we can support direct access to an arbitrary cell of the 2D text in $O(h)$ time, where $h$ is the height of the parse tree, using linear space with respect to the measure. Note that we the balancing procedure of Ganardi et al.~\cite{GJL2021} can be easily adapted, so $h = O(\log N)$, where $N$ is the size of the 2D string.

\begin{definition}
Let $\Mmn$ be a 2D string. A \emph{2-dimensional straight-line program} (2D SLP) for $\Mmn$ is a context-free grammar $(V, \Sigma, R, S)$ that uniquely generates $\Mmn$ and where the definition of the right-hand side of a variable can have the form $$A \rightarrow a,\, A \rightarrow B\ohrz C, \text{ or } A  \rightarrow B \ovrt C,$$ where $a \in \Sigma$, $B, C \in V$. We call these definitions \emph{terminal rules}, \emph{horizontal rules}, and \emph{vertical rules}, respectively. The expansion of a variable is defined as $$\gexp(A) = a,\, \gexp(A) = \gexp(B)\ohrz \gexp(C), \text{ or } \gexp(A) = \gexp (B) \ovrt \gexp(C),$$ respectively.

The measure $g_{2D}(\Mmn)$ is defined as the size of the smallest 2D SLP generating $\Mmn$.  
\end{definition}

It is easy to see that $g_{2D}$ coincides with $g$, when one-dimensional strings are considered. So, from now on we use $g$ to denote $g_{2D}$.

\begin{proposition}
It always holds that $g(\Mmn) = \Omega(\log (mn))$. 
\end{proposition}

\begin{proof}
From the starting variable $S$, each substitution step can double the size of the current 2D string of variables. Hence, an 2D SLP of $g$ rules can produce a 2D string of size at most $2^g$. Therefore, a string of size $N = mn$ needs a grammar of size at least $\log_2 N$.
\end{proof}

\begin{proposition}
The problem of determining if there exists a 2DSLP of size at most $k$ generating a text $\Mmn$ is NP-complete.
\end{proposition}

\begin{proof}Clearly, the problem belongs to NP. Observe that the 1D version of the problem, known to be NP-complete~\cite{Charikar2005}, reduces to the 2D version by considering 1D strings as matrices of size $1 \times n$. 
\end{proof}

\begin{proposition}
Let $\Mmn$ be a 2D string. There exists a data structure using $O(g)$ space that supports direct access to any cell $M[i][j]$ in $O(h)$ time.
\end{proposition}

\begin{proof}
Let $G$ be a 2D SLP generating $\Mmn$. We enrich each rule $A \rightarrow BC$ with the number of rows $m_A$ and the number of columns $n_A$ of its expansion. The procedure to access $\gexp(A)[i][j]$ is detailed in Algorithm~\ref{alg:2D-SLP-direct-access} shown in the Appendix~\ref{sec:access}.
\end{proof}


We can extend 2D SLPs with \emph{run-length rules} to obtain more powerful grammars, retaining most of the nice properties that 2D SLPs have.

\begin{definition}
A \emph{2-dimensional run-length straight-line program} (2D RLSLP) is a 2D SLP that in addition allows special rules of the form

$$A \rightarrow  \ohrz^k B  \text{ and } A \rightarrow  \ovrt^k B$$

for $k > 1$, with their expansions defined as 

\begin{align*}
    \gexp(A) = \gexp( \underbrace{B\ohrz B \ohrz \cdots \ohrz B}_{k \text{ times}}) & \text{ \quad \quad \quad and } & \gexp(A) = \gexp( \underbrace{B\ovrt B \ovrt \cdots \ovrt B}_{k \text{ times}})\\
\end{align*}
respectively. The measure $g_{rl2D}(\Mmn)$ is defined as the size of the smallest 2D RLSLP generating a text $\Mmn$.    
\end{definition}

\begin{proposition}It always hold that $g_{rl} = O(g)$. Moreover, there are string families where $g_{rl} = o(g)$.
\end{proposition}

\begin{proof}First claim is trivial by definition. Second claim is derived by considering the family of $1 \times n$ matrices with first row $\one^n$ as 1D strings.
\end{proof}

Analogously to the strategy described in~\cite{BilleLRSSW15}, the following result can be obtained. The full procedure is described in Algorithm~\ref{alg:2D-SLP-direct-access} in the the Appendix~\ref{sec:access}.

\begin{proposition}
Let $\Mmn$ be a 2D string. There exists a data structure using $O(g_{rl})$ space that supports direct access to any cell $M[i][j]$ in $O(h)$ time.
\end{proposition}


\section{Macro Schemes for 2D Strings}

As natural extensions to 2D strings, $\delta$ and $\gamma$ measures also inherit the fact that $\delta$ is unreachable and $\gamma$ is unknown to be reachable. For this reason, also in the context of 2D strings, it can be interesting and useful to explore repetitiveness measures based on extensions to 2D strings of bidirectional macro schemes. In this section, we introduce such a notion. Then, in Section~\ref{sec:differences} we relate 2D macro schemes to $\delta$ and $\gamma$.

\begin{definition}Let $M_{m \times n}$ be a 2D string. A \emph{2D macro scheme} for $\Mmn$ is any factorization of $\Mmn$ into a set of phrases such that any phrase is either a square of dimension $1 \times 1$ called an \emph{explicit symbol/phrase}, or has a source in $\Mmn$ starting at a different position (which we can specify by its top-left corner). For a 2D macro scheme to be  \emph{valid} or \emph{decodifiable}, there must exist a function $\map: ([1 \dd m] \times [1\dd n])  \cup \{\bot\} \rightarrow ([1 \dd m] \times [1\dd n]) \cup \{\bot\}$ such that: i) $\map(\bot) = \bot$, and if $M[i][j]$ is an explicit symbol, then $\map(i,j) = \bot$; ii) for each non-explicit phrase $M[i_1\dd j_1][i_2\dd j_2]$, it must hold that $\map(i_1+t_1, i_2 + t_2) = \map(i_1, i_2) + (t_1, t_2)$ for $(t_1, t_2) \in [0\dd j_1-i_1] \times [0\dd j_2 - i_2]$; iii) for each $(i, j) \in [1\dd m] \times [1 \dd n]$ there exists $k > 0$ such that $\map^k(i, j) = \bot$.

The measure $b_{2D}$ is defined as the size of the smallest valid 2D macro scheme for $\Mmn$.
\end{definition}

\begin{example}
\label{ex:b_O1}Let $I_{n}$ be the identity matrix of dimension $n \times n$ over $\{\zero, \one\}$. We construct a macro scheme $\{X_1, X_2, X_3, X_4, X_5, X_6\}$ for this string, where: i) $X_1 = I_n[1][1]$ is an explicit symbol (the $\one$ in the top-left corner); ii) $X_2 = I_n[1][2]$ is an explicit symbol; $X_3 = I_n[2][1]$ is an explicit symbol; $X_4 = I_n[1][3\dd n]$ is a phrase with source $(1,2)$; $X_5 = I_n[3\dd n][1]$ is a phrase with source $(2,1)$; and $X_6 = I_n[2 \dd n][2\dd n]$ is a phrase with source $(1,1)$. The underlying function $\map$ can be defined as $\map(1,1) = \map(1, 2) = \map(2,1) = \bot$, $\map(1,j) = (1, j-1)$ for $j \in [3\dd n]$, $\map(i, 1) = (i-1, 1)$ for $i \in [3\dd n]$, and $\map(i,j) = (i-1, j-1)$ for $i, j \in [2\dd n] \times [2\dd n]$. One can see that $\map^n(i,j) = \bot$ for each $i$ and $j$. Hence, the macro scheme is valid and it holds that $b_{2D} \le 6$ in the family $I_{n}$. Figure~\ref{fig:fig1} shows this macro scheme for the string $I_7$.
\end{example}

\begin{figure}[ht]
\centering
\begin{tikzpicture}
    [
        box/.style={rectangle,draw=black,thick, minimum size=1cm},
    ]

\foreach \x in {0,1,...,6}{
    \foreach \y in {0,-1,...,-6}{
         \ifthenelse{\x=-\y}{\node at (\x,\y){\one};}{\node at (\x,\y){\zero};}
}}
\draw (-0.5,0.5) rectangle (0.5,-0.5);
\draw (0.5,-0.5) rectangle (6.5,-6.5);
\draw (-0.5,0.5) rectangle (6.5,-6.5);
\draw (-0.5,-1.5) rectangle (0.5,-1.5);
\draw (0.5,0.5) rectangle (1.5,-0.5);
\draw[dashed, ->] (0.8,-0.8) to[bend right] (0.2,-0.2);
\draw[dashed, ->] (0,-1.7) to[bend right] (0,-1.3);
\draw[dashed, ->] (1.7,0) to[bend right] (1.3, 0);
\end{tikzpicture}
\caption{Macro scheme with 6 phrases for the 2D string $I_7$. Each phrase points to the source from where it is directly copied.}
\label{fig:fig1}
\end{figure}

Analogously to the measures seen in previous sections, $b_{2D}$ reduces to definition of $b$ on 1D strings when either $m = 1$ or $n=1$.
Hence, also in case of macro schemes, we will use $b$ to denote $b_{2D}$.
In the following, we show that the relationship with the measures $g_{rl}$ and $g$ are preserved.

\begin{proposition}
The problem of determining if there exists a valid 2D macro scheme of size at most $k$ for a text $\Mmn$ is NP-complete.
\end{proposition}

\begin{proof}
The 1D version of the problem, which is known to be NP-complete~\cite{Gallant1982}, reduces to the 2D version of the problem in constant time.
\end{proof}

\begin{proposition}It always holds that $b = O(g_{rl})$. Moreover, there are string families where $b = o(g_{rl})$.
\end{proposition}

\begin{proof}We show how to construct a macro scheme from a 2D RLSLP, representing the same string and having the same asymptotic size. Let $G$ be a 2D RLSLP generating $\Mmn$ and consider its grammar tree. The first occurrence of each variable $B$ that expands to a single symbol at position $M[i][j]$, becomes an explicit phrase of the parsing at that position. For each occurrence of a variable $B$ that is not the first one in this tree, we create a phrase --exactly where this occurrence of $\gexp(B)$ should be in $M$-- that maps to the expansion of the first occurrence of $B$ in $M$. For a a rule $A \rightarrow \ohrz^kB$, we construct at most two phrases: one phrase for the leftmost $B$ of $\ohrz^kB$ pointing to the first occurrence of $B$, or if the expansion of this $B$ is a single symbol (or no phrase if this $B$ is the first occurrence but its expansion is not a single symbol), and other phrase for $\ohrz^{k-1}B$ pointing to the expansion of the occurrence of the $B$ before in $M$. We do analogously, for 2D vertical run-length rules. It is easy to see that this parsing is decodifiable, and that its size is bounded by the size of the grammar tree, that is, it is $O(g_{rl})$. For a family where $b = o(g_{rl})$, this holds for the 1D version~\cite{NOP20}.
\end{proof}

\section{Differences Between the 1D and the 2D Setting}\label{sec:differences}
On the 1-dimensional context, for each string $S\in\Sigma^*$ it holds that $\delta= O(\gamma) = O(b) = O(g_{rl})= O(g)$.
In particular, while $\delta \leq \gamma$ and $b \leq g_{rl} \leq g$ truly rely on their definitions, the missing link between $\gamma$ and $b$ has been proved by Kempa and Prezza by showing how any macro scheme of size $b$ induces a suitable string attractor with at most $2b$ positions~\cite{KP18}.
Later, Bannai et al.~\cite{BFIKMN21} showed that for 1D strings there is a separation between $\gamma$ and $b$ by using the family of Thue-Morse words, that is a string family for which $\gamma=O(1)$ and $b = \Theta(\log n)$.

In the previous sections, we have showed that on the 2D setting the same relationships between $\delta$ and $\gamma$ holds, as well as the one between $b$, $g_{rl}$, and $g$.
However, unlike the 1D settings, we can have 2D strings for which the measure $b$ is asymptotically smaller than $\gamma$.


\begin{proposition}\label{prop:b_ogamma}
There exists a 2D string family where $\gamma = \Omega(b \sqrt{N})$, where $N$ is the size of the 2D string.
\end{proposition}

\begin{proof}As shown in the proof of Proposition~\ref{prop:gamma_odelta}, $\gamma = \Omega(n)$ in the family of identity matrices $I_{n}$. 
On the other hand, in Example~\ref{ex:b_O1} we showed how to construct a macro scheme with only 6 phrases for the same family of strings, i.e. $b = O(1)$.
The claim follows since for every identity matrix $I_n$ it holds that $n = \sqrt{N}$.

\end{proof}

It follows that, when considering a 2D setting, the measure $b$ can be much better than $\gamma$. Hence, both measures are uncomparable with each other.
A follow up question is whether the relationship between the measures $\delta$ and $b$ on the 2-dimensional context is preserved.
As the following proposition shows, not only the measure $b$ can be asymptotically smaller than $\delta$, but it can be asymptotically smaller than $\delta_\square$ too.

\begin{proposition}\label{prop:b_lt_delta_square}
There exists a 2D string family where $\delta_\square = \Omega(b\sqrt[4]{N})$, where $N$ is the size of the 2D string.
\end{proposition}

\begin{proof}
Let $k > 3$. Let $\mathcal{F}(k)$ be a set containing all the 2D strings of dimensions $k \times k$ where: i) exactly two cells in distinct rows contain a $\one$; ii) the $\one$ on the row above cannot be more to the right than the $\one$ on the row below; iii) all the remaining cells contain $\zero$'s. There are $\binom{k}{2}$ ways to choose two distinct rows to verify i), and then $k^2$ ways to choose which cells in these rows contain the $\one$'s.
However, only $k(k+1)/2$ of these 2D strings satisfy condition ii).
Hence, there exist exactly $k^2(k-1)(k+1)/4$ of such strings in $\mathcal{F}(k)$.

Let us now construct a 2D string $A_k$ containing all the $k \times k$ strings in $\mathcal{F}(k)$ as substrings.
Let $B(i,j)$ be the $k \times k$ string containing a $\one$ in position $(1,1)$, and another $\one$ in position $(i,j)$.
The matrix $A_k$ is defined as follows: for each $i\in [2\dd k]$ take the $k \times k^2$ substring containing only $\zero$'s and append below the $k \times k^2$ matrix  $B(i,1) \ohrz B(i, 2) \ohrz \dots \ohrz B(i,k)$.
Then, concatenate all these matrices from top to bottom.
Finally, append to the left and to the right a $(k-1)2k \times k$ substring containing only $\zero$'s.
Schematically, the matrix $A_k$ has the following form:
$$\begin{array}{lllll}
0_{k\times k} & 0_{k \times k} & \dots & 0_{k\times k} & 0_{k \times k}\\
0_{k\times k} & B(2,1) & \dots & B(2, k) & 0_{k \times k}\\ 
\vdots & \vdots & \vdots & \vdots & \vdots \\
0_{k\times k} & 0_{k \times k} & \dots & 0_{k\times k} & 0_{k \times k}\\
0_{k\times k} & B(k,1) & \dots & B(k, k) & 0_{k \times k}\\
\end{array}$$
One can see that we can move a $k \times k$ window containing both $\one$'s of some matrix $B(i, j)$ for some $i$ and $j$ to find any string in $\mathcal{F}(k)$ as a substring of $A_k$.
Thus, it holds that $\delta_\square(A_k) = \Omega(k^2)$.
Let $i \in [2\dd k]$ and $j \in [1\dd k]$.
Moreover, the size of $A_k$ is $N = m\times n = 2k(k-1) \times k(k+2) = \Theta(k^4)$. 

Now we show how to construct a valid macro scheme for $A_k$.
The intuition for the macro scheme is to first create phrases for the rectangles containing only $\zero$'s surrounding the central part of $A_k$.
Then, we observe that the submatrices $B(i,1) \ohrz B(i, 2) \ohrz \dots \ohrz B(i,k)$ contain only 3 types of rows: 1) $\zero^{k^2}$; 2) $(\one\zero^{k-1})^k$; or 3) $(\one\zero^k)^{k-1}\one$.
Thus, we can use $O(1)$ phrases for the first occurrence (at the top) of rows of type 2) and 3), and then use them as a reference for the other occurrences.
More in detail, the macro scheme contains the phrases $\{X_1, \dots, X_{12}\}$ where: 
i) $X_1 = A_k[1][1]$ is an explicit phrase containing a $\zero$;
ii) $X_2 = A_k[2\dd k][1]$ is a phrase with source $(1,1)$;
iii) $X_3 = A_k[1\dd k][2\dd n]$ is a phrase with source $(1,1)$; 
iv) $X_4 = A_k[k+1 \dd m][1\dd k]$ is a phrase with source $(1,1)$;
v) $X_5 = A_k[k+1 \dd m][k(k+1)+1\dd n]$ is a phrase with source $(1,1)$;
vi) $X_6 = A_k[k+1][k+1]$ is a explicit phrase containing a $\one$;
vii) $X_7 = A_k[k+1][k+2\dd 2k]$ is a phrase with source $(1,1)$;
viii) $X_8 = A_k[k+1][2k+1\dd k(k+1)]$ is a phrase with source $(k+1,k+1)$;
ix) $X_9 = A_k[k+2][k+1]$ is a explicit phrase containing a $\one$;
x) $X_{10} = A_k[k+2][k+2\dd 2k+1]$ is a phrase with source $(1,1)$;
xi) $X_{11} = A_k[k+2][2k+2\dd k(k+1)]$ is a phrase with source $(k+2,k+1)$;
xii) $X_{12} = A_k[k+3\dd 2k][k+1\dd k(k+1)]$ is a phrase with source $(1,1)$.

Observe that the remaining phrases refer to the matrix $A_k[2k+1\dd m][k+1\dd k(k+1)]$, with rows of the type 1), 2), and 3) described above.
Hence, each range of consecutive rows of type 1) (i.e. of all $\zero$'s) can be copied from the (biggest) block of consecutive rows of $\zero$'s starting at the beginning of phrase $X_{12}$, and we have $2(k-2)$ of such phrases.
For each row of type 2) we use a single phrase pointing to the beginning of its first occurrence, at the beginning of phrase $X_6$, and we have exactly $k-2$ of such rows.
Analogously, for each of the $k-2$ rows of type 3), we use a single phrase pointing to the beginning of phrase $X_9$.
Thus, the total size of the macro scheme built is indeed $12 + 4(k-2) = O(k)$. 
Hence, we proved that $\delta_\square = \Omega(bk)$. As $k = \Theta(\sqrt[4]{N})$, the claim follows.
\end{proof}

\begin{figure}[ht]
\centering
\tiny
 \begin{tikzpicture}
        \matrix [matrix of math nodes,left delimiter=(,right delimiter=)] (m)
        {
            \zero & \zero & \zero & \zero & \zero & \zero  & \zero & \zero  & \zero & \zero  & \zero & \zero &  \zero & \zero & \zero & \zero & \zero & \zero  & \zero & \zero  & \zero & \zero  & \zero & \zero \\
            \zero & \zero & \zero & \zero & \zero & \zero  & \zero & \zero  & \zero & \zero  & \zero & \zero &  \zero & \zero & \zero & \zero & \zero & \zero  & \zero & \zero  & \zero & \zero  & \zero & \zero\\
            \zero & \zero & \zero & \zero & \zero & \zero  & \zero & \zero  & \zero & \zero  & \zero & \zero &  \zero & \zero & \zero & \zero & \zero & \zero  & \zero & \zero  & \zero & \zero  & \zero & \zero \\
            \zero & \zero & \zero & \zero & \zero & \zero  & \zero & \zero  & \zero & \zero  & \zero & \zero &  \zero & \zero & \zero & \zero & \zero & \zero  & \zero & \zero  & \zero & \zero  & \zero & \zero \\
            \zero & \zero & \zero & \zero & \one & \zero  & \zero & \zero  & \one & \zero  & \zero & \zero & \one & \zero  & \zero & \zero & \one & \zero  & \zero & \zero & \zero & \zero & \zero & \zero \\
            \zero & \zero & \zero & \zero & \one & \zero  & \zero & \zero  & \zero & \one  & \zero & \zero & \zero & \zero & \one & \zero & \zero & \zero & \zero & \one & \zero & \zero & \zero & \zero\\
            \zero & \zero & \zero & \zero & \zero & \zero  & \zero & \zero  & \zero & \zero  & \zero & \zero &  \zero & \zero & \zero & \zero & \zero & \zero  & \zero & \zero  & \zero & \zero  & \zero & \zero \\
            \zero & \zero & \zero & \zero & \zero & \zero  & \zero & \zero  & \zero & \zero  & \zero & \zero &  \zero & \zero & \zero & \zero & \zero & \zero  & \zero & \zero  & \zero & \zero  & \zero & \zero \\
            \zero & \zero & \zero & \zero & \zero & \zero  & \zero & \zero  & \zero & \zero  & \zero & \zero &  \zero & \zero & \zero & \zero & \zero & \zero  & \zero & \zero  & \zero & \zero  & \zero & \zero \\
            \zero & \zero & \zero & \zero & \zero & \zero  & \zero & \zero  & \zero & \zero  & \zero & \zero &  \zero & \zero & \zero & \zero & \zero & \zero  & \zero & \zero  & \zero & \zero  & \zero & \zero\\
            \zero & \zero & \zero & \zero & \zero & \zero  & \zero & \zero  & \zero & \zero  & \zero & \zero &  \zero & \zero & \zero & \zero & \zero & \zero  & \zero & \zero  & \zero & \zero  & \zero & \zero \\
            \zero & \zero & \zero & \zero & \zero & \zero  & \zero & \zero  & \zero & \zero  & \zero & \zero &  \zero & \zero & \zero & \zero & \zero & \zero  & \zero & \zero  & \zero & \zero  & \zero & \zero \\
            \zero & \zero & \zero & \zero & \one & \zero  & \zero & \zero  & \one & \zero  & \zero & \zero & \one & \zero  & \zero & \zero & \one & \zero  & \zero & \zero & \zero & \zero & \zero & \zero \\
             \zero & \zero & \zero & \zero & \zero & \zero  & \zero & \zero  & \zero & \zero  & \zero & \zero &  \zero & \zero & \zero & \zero & \zero & \zero  & \zero & \zero  & \zero & \zero  & \zero & \zero \\
            \zero & \zero & \zero & \zero & \one & \zero  & \zero & \zero  & \zero & \one  & \zero & \zero & \zero & \zero & \one & \zero & \zero & \zero & \zero & \one & \zero & \zero & \zero & \zero\\
             \zero & \zero & \zero & \zero & \zero & \zero  & \zero & \zero  & \zero & \zero  & \zero & \zero &  \zero & \zero & \zero & \zero & \zero & \zero  & \zero & \zero  & \zero & \zero  & \zero & \zero \\
                         \zero & \zero & \zero & \zero & \zero & \zero  & \zero & \zero  & \zero & \zero  & \zero & \zero &  \zero & \zero & \zero & \zero & \zero & \zero  & \zero & \zero  & \zero & \zero  & \zero & \zero \\
            \zero & \zero & \zero & \zero & \zero & \zero  & \zero & \zero  & \zero & \zero  & \zero & \zero &  \zero & \zero & \zero & \zero & \zero & \zero  & \zero & \zero  & \zero & \zero  & \zero & \zero\\
            \zero & \zero & \zero & \zero & \zero & \zero  & \zero & \zero  & \zero & \zero  & \zero & \zero &  \zero & \zero & \zero & \zero & \zero & \zero  & \zero & \zero  & \zero & \zero  & \zero & \zero \\
            \zero & \zero & \zero & \zero & \zero & \zero  & \zero & \zero  & \zero & \zero  & \zero & \zero &  \zero & \zero & \zero & \zero & \zero & \zero  & \zero & \zero  & \zero & \zero  & \zero & \zero \\
            \zero & \zero & \zero & \zero & \one & \zero  & \zero & \zero  & \one & \zero  & \zero & \zero & \one & \zero  & \zero & \zero & \one & \zero  & \zero & \zero & \zero & \zero & \zero & \zero \\
             \zero & \zero & \zero & \zero & \zero & \zero  & \zero & \zero  & \zero & \zero  & \zero & \zero &  \zero & \zero & \zero & \zero & \zero & \zero  & \zero & \zero  & \zero & \zero  & \zero & \zero \\
              \zero & \zero & \zero & \zero & \zero & \zero  & \zero & \zero  & \zero & \zero  & \zero & \zero &  \zero & \zero & \zero & \zero & \zero & \zero  & \zero & \zero  & \zero & \zero  & \zero & \zero \\
            \zero & \zero & \zero & \zero & \one & \zero  & \zero & \zero  & \zero & \one  & \zero & \zero & \zero & \zero & \one & \zero & \zero & \zero & \zero & \one & \zero & \zero & \zero & \zero\\
        };  
        \draw[thick] (m-5-5.north west) -- (m-5-20.north east) -- (m-8-20.south east) -- (m-8-5.south west) -- (m-5-5.north west);
        \draw[thick] (m-13-5.north west) -- (m-13-20.north east) -- (m-16-20.south east) -- (m-16-5.south west) -- (m-13-5.north west);
        \draw[thick] (m-21-5.north west) -- (m-21-20.north east) -- (m-24-20.south east) -- (m-24-5.south west) -- (m-21-5.north west);
        \draw[thick] (m-5-9.north west) -- (m-8-9.south west);
        \draw[thick] (m-5-13.north west) -- (m-8-13.south west);
        \draw[thick] (m-5-17.north west) -- (m-8-17.south west);
        \draw[thick] (m-13-9.north west) -- (m-16-9.south west);
        \draw[thick] (m-13-13.north west) -- (m-16-13.south west);
        \draw[thick] (m-13-17.north west) -- (m-16-17.south west);
        \draw[thick] (m-21-9.north west) -- (m-24-9.south west);
        \draw[thick] (m-21-13.north west) -- (m-24-13.south west);
        \draw[thick] (m-21-17.north west) -- (m-24-17.south west);
        \draw[thick, dashed] (m-3-2.north west) -- (m-3-5.north east) -- (m-6-5.south east) -- (m-6-2.south west) -- (m-3-2.north west);
        \draw[thick, dashed] (m-12-8.north west) -- (m-12-11.north east) -- (m-15-11.south east) -- (m-15-8.south west) -- (m-12-8.north west);
        \draw[dotted] (m-1-5.north west) -- (m-24-5.north west);
        \draw[dotted] (m-1-20.north east) -- (m-24-20.north east);
        \draw[dotted] (m-5-1.north west) -- (m-5-24.north east);
    \end{tikzpicture}\hspace{0.5cm}
    \begin{tikzpicture}
        \matrix [matrix of math nodes] (m)
        {
            \zero & \zero & \zero & \zero & \zero & \zero  & \zero & \zero  & \zero & \zero  & \zero & \zero &  \zero & \zero & \zero & \zero & \zero & \zero  & \zero & \zero  & \zero & \zero  & \zero & \zero \\
            \zero & \zero & \zero & \zero & \zero & \zero  & \zero & \zero  & \zero & \zero  & \zero & \zero &  \zero & \zero & \zero & \zero & \zero & \zero  & \zero & \zero  & \zero & \zero  & \zero & \zero\\
            \zero & \zero & \zero & \zero & \zero & \zero  & \zero & \zero  & \zero & \zero  & \zero & \zero &  \zero & \zero & \zero & \zero & \zero & \zero  & \zero & \zero  & \zero & \zero  & \zero & \zero \\
            \zero & \zero & \zero & \zero & \zero & \zero  & \zero & \zero  & \zero & \zero  & \zero & \zero &  \zero & \zero & \zero & \zero & \zero & \zero  & \zero & \zero  & \zero & \zero  & \zero & \zero \\
            \zero & \zero & \zero & \zero & \one & \zero  & \zero & \zero  & \one & \zero  & \zero & \zero & \one & \zero  & \zero & \zero & \one & \zero  & \zero & \zero & \zero & \zero & \zero & \zero \\
            \zero & \zero & \zero & \zero & \one & \zero  & \zero & \zero  & \zero & \one  & \zero & \zero & \zero & \zero & \one & \zero & \zero & \zero & \zero & \one & \zero & \zero & \zero & \zero\\
            \zero & \zero & \zero & \zero & \zero & \zero  & \zero & \zero  & \zero & \zero  & \zero & \zero &  \zero & \zero & \zero & \zero & \zero & \zero  & \zero & \zero  & \zero & \zero  & \zero & \zero \\
            \zero & \zero & \zero & \zero & \zero & \zero  & \zero & \zero  & \zero & \zero  & \zero & \zero &  \zero & \zero & \zero & \zero & \zero & \zero  & \zero & \zero  & \zero & \zero  & \zero & \zero \\
            \zero & \zero & \zero & \zero & \zero & \zero  & \zero & \zero  & \zero & \zero  & \zero & \zero &  \zero & \zero & \zero & \zero & \zero & \zero  & \zero & \zero  & \zero & \zero  & \zero & \zero \\
            \zero & \zero & \zero & \zero & \zero & \zero  & \zero & \zero  & \zero & \zero  & \zero & \zero &  \zero & \zero & \zero & \zero & \zero & \zero  & \zero & \zero  & \zero & \zero  & \zero & \zero\\
            \zero & \zero & \zero & \zero & \zero & \zero  & \zero & \zero  & \zero & \zero  & \zero & \zero &  \zero & \zero & \zero & \zero & \zero & \zero  & \zero & \zero  & \zero & \zero  & \zero & \zero \\
            \zero & \zero & \zero & \zero & \zero & \zero  & \zero & \zero  & \zero & \zero  & \zero & \zero &  \zero & \zero & \zero & \zero & \zero & \zero  & \zero & \zero  & \zero & \zero  & \zero & \zero \\
            \zero & \zero & \zero & \zero & \one & \zero  & \zero & \zero  & \one & \zero  & \zero & \zero & \one & \zero  & \zero & \zero & \one & \zero  & \zero & \zero & \zero & \zero & \zero & \zero \\
             \zero & \zero & \zero & \zero & \zero & \zero  & \zero & \zero  & \zero & \zero  & \zero & \zero &  \zero & \zero & \zero & \zero & \zero & \zero  & \zero & \zero  & \zero & \zero  & \zero & \zero \\
            \zero & \zero & \zero & \zero & \one & \zero  & \zero & \zero  & \zero & \one  & \zero & \zero & \zero & \zero & \one & \zero & \zero & \zero & \zero & \one & \zero & \zero & \zero & \zero\\
             \zero & \zero & \zero & \zero & \zero & \zero  & \zero & \zero  & \zero & \zero  & \zero & \zero &  \zero & \zero & \zero & \zero & \zero & \zero  & \zero & \zero  & \zero & \zero  & \zero & \zero \\
                         \zero & \zero & \zero & \zero & \zero & \zero  & \zero & \zero  & \zero & \zero  & \zero & \zero &  \zero & \zero & \zero & \zero & \zero & \zero  & \zero & \zero  & \zero & \zero  & \zero & \zero \\
            \zero & \zero & \zero & \zero & \zero & \zero  & \zero & \zero  & \zero & \zero  & \zero & \zero &  \zero & \zero & \zero & \zero & \zero & \zero  & \zero & \zero  & \zero & \zero  & \zero & \zero\\
            \zero & \zero & \zero & \zero & \zero & \zero  & \zero & \zero  & \zero & \zero  & \zero & \zero &  \zero & \zero & \zero & \zero & \zero & \zero  & \zero & \zero  & \zero & \zero  & \zero & \zero \\
            \zero & \zero & \zero & \zero & \zero & \zero  & \zero & \zero  & \zero & \zero  & \zero & \zero &  \zero & \zero & \zero & \zero & \zero & \zero  & \zero & \zero  & \zero & \zero  & \zero & \zero \\
            \zero & \zero & \zero & \zero & \one & \zero  & \zero & \zero  & \one & \zero  & \zero & \zero & \one & \zero  & \zero & \zero & \one & \zero  & \zero & \zero & \zero & \zero & \zero & \zero \\
             \zero & \zero & \zero & \zero & \zero & \zero  & \zero & \zero  & \zero & \zero  & \zero & \zero &  \zero & \zero & \zero & \zero & \zero & \zero  & \zero & \zero  & \zero & \zero  & \zero & \zero \\
              \zero & \zero & \zero & \zero & \zero & \zero  & \zero & \zero  & \zero & \zero  & \zero & \zero &  \zero & \zero & \zero & \zero & \zero & \zero  & \zero & \zero  & \zero & \zero  & \zero & \zero \\
            \zero & \zero & \zero & \zero & \one & \zero  & \zero & \zero  & \zero & \one  & \zero & \zero & \zero & \zero & \one & \zero & \zero & \zero & \zero & \one & \zero & \zero & \zero & \zero\\
        }; 
        \draw (m-1-1.north west) -- (m-1-1.north east) -- (m-1-1.south east) -- (m-1-1.south west) -- (m-1-1.north west);
        \draw (m-1-1.south east) -- (m-4-1.south east) -- (m-4-1.south west) -- (m-2-1.north west);
        \draw (m-1-1.north east) -- (m-1-24.north east) -- (m-4-24.south east) -- (m-4-1.south east);
        \draw (m-4-4.south east) -- (m-24-4.south east) -- (m-24-1.south west) -- (m-5-1.north west);
        \draw  (m-4-24.south east) -- (m-24-24.south east) -- (m-24-21.south west) -- (m-4-21.south west);
        \draw  (m-4-5.south east) -- (m-5-5.south east) -- (m-5-4.south east);
        \draw  (m-4-8.south east) -- (m-5-8.south east) -- (m-5-5.south east);
        \draw  (m-5-8.south east) -- (m-5-21.south west);
        \draw  (m-5-5.south east) -- (m-6-5.south east) -- (m-6-4.south east);
        \draw  (m-5-9.south east) -- (m-6-9.south east) -- (m-6-5.south east);
        \draw  (m-6-9.south east) -- (m-6-21.south west);
        \draw  (m-8-4.south east) -- (m-8-21.south west);
        \draw  (m-12-4.south east) -- (m-12-21.south west);
        \draw  (m-13-4.south east) -- (m-13-21.south west);
        \draw  (m-13-4.south east) -- (m-13-21.south west);
        \draw  (m-14-4.south east) -- (m-14-21.south west);
        \draw  (m-15-4.south east) -- (m-15-21.south west);
        \draw  (m-20-4.south east) -- (m-20-21.south west);
        \draw  (m-21-4.south east) -- (m-21-21.south west);
        \draw  (m-23-4.south east) -- (m-23-21.south west);
        \draw  (m-24-4.south east) -- (m-24-21.south west);
        %
        \draw[->]  (m-13-5) to [bend left=90] (m-5-5);
        %
        \draw[->]  (m-15-5) to [bend left=90] (m-6-5);
        %
        \draw[->]  (m-9-5) to [bend left=90] (m-7-5);
        \draw[->]  (m-14-5) to [bend left=90] (m-7-5);
    \end{tikzpicture}
    \caption{Matrix $A_k$ for $k = 4$. To the left, we highlight the substrings $B(i,j)$, and show how to obtain other strings in $\mathcal{F}(k)$ by moving a $4 \times 4$ window. To the right, we show the macro scheme described in Proposition~\ref{prop:b_lt_delta_square}, which is formed by exactly $4(k+1)$ phrases. Each arrow show the source from where the phrases are copied.}
    \label{fig:matrix}
\end{figure}
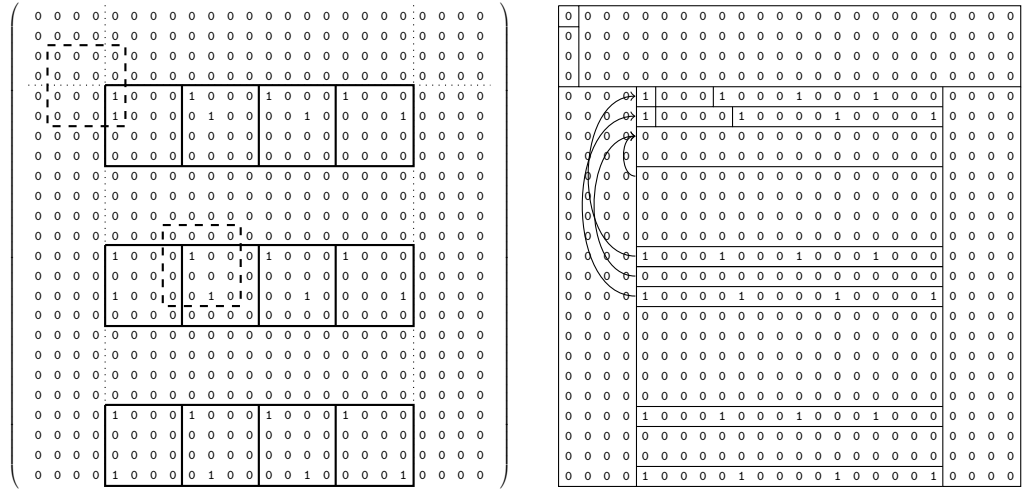

An example on $A_4$ can be seen in Figure~\ref{fig:matrix}.
Observe that we can obtain the same result if we compare $b$ and $\delta_\square$ on a square matrix, i.e. with the same setting from~\cite{CarfagnaManzini2023}.
In fact, since the measure $\delta_\square$ is monotone, we can append columns of $\zero$'s to the left or to the right of $A_k$ until we obtain a square preserving the lower bound $\delta_\square=\Omega(k^2)$, while a suitable macro scheme may require at most $O(1)$ new phrases.

As $\delta_\square\leq \delta$ for all 2D strings, we derive the following corollary.

\begin{corollary}
    There exists a 2D string family where $b=o(\delta)$.
\end{corollary}

Even 2D SLPs can be noticeable smaller than $\delta_\square$.

\begin{proposition}
There exists a 2D string family where $\delta = \Omega(g\sqrt[4]{N}/\log N)$, where $N$ is the size of the 2D string.
\end{proposition}

\begin{proof}We show that in the same family of strings $A_k$ of Proposition~\ref{prop:b_lt_delta_square}, it holds that $g = O(\sqrt[4]{N}\log N)$ whereas $\delta_\square = \Omega(\sqrt{N})$. Notice that it is always possible to generate a string $\zero_{k_1\times k_2}$ with a 2D SLP of size $\Theta(\log (k_1k_2))$. Hence, we can obtain 2D SLPs generating all the phrases of this type in the macro scheme of Proposition~\ref{prop:b_lt_delta_square}, having total size $O(\log N^b) = O(b\log N) = O(\sqrt[4]{N}\log N)$. Similarly, to generate $(\one\zero^{k-1})^k$ and $(\one\zero^k)^{k-1}\one$ we need 2D SLPs whose size sums to $O(\log k) = O(\log N)$. Finally, we need $O(\sqrt[4]{N}\log N)$ new rules to merge all the 2D SLPs described before. The total size of the 2D SLP is indeed $O(\sqrt[4]{N}\log N)$. Thus, the result follows.
\end{proof}

\begin{corollary}
There exists a string family where $g = o(\delta)$.    
\end{corollary}

We point out that some of the bounds on the gaps showed in the previous propositions are not tight.
Nevertheless, we believe that the uncomparability of $\delta$ (or $\delta_\square$) and $\gamma$ (or $\gamma_\square$) with $b$ is enough to show their weaknesses as measures when it comes to consider bi-dimensional strings.
We defer other examples with better bounds for an extended version of this work.


\section{Conclusions and Future Work}

In this work, we have proposed extensions to many popular repetitiveness measures to make them suitable for two-dimensional strings. 

We have shown that the definitions of $\delta$ and $\gamma$ lose most of the properties that makes them meaningful in the 1D setting. In particular, the uncomparability with $b$ from Section~\ref{sec:differences} questions the usage of $\delta$ (or $\delta_\square$) and $\gamma$ (or $\gamma_\square$) as measures of repetitiveness when it comes to consider two-dimensional strings.
We wonder if there exists a repetitiveness measure comparable to $\delta$ on the 1D setting in terms of efficiency, and such that it preserves most of its properties and relationship with other measures when considering 2D strings.

On the other hand, the extensions of reachable and accessible measures that we have proposed maintain most of their usefulness and could have practical applications if studied more: analogously to what happens on one-dimensional strings, through (RL)SLPs one can access and extract 2D factors in $O(g)$ (or $O(g_{rl})$) space.
In particular, we wonder if it is possible to efficiently construct 2D (RL)SLPs with size bounded in terms of $b$ (or better comparable measures, if any).

It should be noticed that most of our definitions and results can be easily generalized to hold for $d$-dimensional strings: as a 2D string $\Mmn$ can be seen as the concatenation of $n$ 1D strings of size $m$ over the second dimension, a 3D string $M_{m\times n \times \ell}$ can be seen as the concatenation of $\ell$ 2D strings of size $m\times n$ over a third dimension, hence in the same fashion we can recursively define every $d$D string.
It follows that the $d$-dimensional versions of the substring complexity and the measures $\delta$ and $\gamma$ consider $d$-dimensional factors, while a $d$-dimensional macro scheme copies $d$-dimensional factors and store phrases when the size of these factors is 1 (i.e., length 1 on all the dimensions).
Analogously, for a $d$-dimensional (RL)SLPs we further consider rules that concatenate non-terminals symbols over all the $d$ dimensions.
Since for all $d$-dimensional string  with $d > 2$ and lengths 1 over $d-2$ dimensions it holds that the $d$-dimensional versions of such measures are reduced to the 2D setting, all the results obtained in this paper hold for all $d>2$.
We want also to remark that any symmetry or rotation over one of the $d$ dimensions does not affect any of the measures here presented.
%
%
%
\bibliography{bibliography}

\newpage
\appendix

\section{Differences between 2D measures and 1D measures on linearized matrices}\label{sec:linearization}

In this section, we show some examples for which the order for the measures $\gamma$, $b$, and $g$ on 2D strings differ if compared with the same measures computed on some linearization of the matrix.


The following example shows a family of 2D strings for which the $\gamma$ measure applied to the string obtained by using the $\rowlin$ linearization has lower order of magnitude than the 2D version of $\gamma$.

\begin{example}\label{ex:gamma_linearization}
Consider the family of identity matrices $I_n$. Each row in $I_n$ is a unique submatrix of size $1 \times n$. As they are disjoint submatrices, each row needs at least one attractor position. Hence $\gamma(I_n) = \Omega(n)$. 
On the other hand, recall the definition of $\rowlin$ from Example~\ref{ex:delta_row_linearization}. One can observe that $\rowlin(I_n) = (\one\zero^n)^{n-1}\one$. Therefore, the set  $\{1,2,n+1\}$ is a string attractor for $\rowlin(I_n)$, i.e. $\gamma(\rowlin(I_n)) = O(1)$.
\end{example}


On the other hand, Example~\ref{ex:bglinearization} shows a family of 2D strings for which both the measures $g$ and $b$ are more efficient when considering the 2D setting instead of the linearization by row.

\begin{example}\label{ex:bglinearization}Consider the family obtained by taking the matrices $I_{n-1}$ where $n-1 = 2^k$ for some $k$, then appending a column of $\one$'s on the right and a row of $\zero$'s at the bottom (in that order). 
We have already shown in Example~\ref{ex:delta_row_linearization} that $\delta(\rowlin(\Mmn)) = \Omega(n)$ in this string family. Hence, this also holds for $g(\rowlin(\Mmn))$. On the other hand, a 2D SLP for a string in this family is obtained given by the following rules:

\begin{enumerate}
\item $S \rightarrow S' \ominus \zero_{1 \times n}$
\item $S' \rightarrow I_n \ohrz \one_{n \times 1}$
\item $I_k \rightarrow (I_{k/2} \ohrz 0_{k/2}) \ominus (0_{k/2} \ohrz I_{k/2})$
\item $\zero_k \rightarrow  (\zero_{k/2} \ohrz \zero_{k/2}) \ominus (\zero_{k/2} \ohrz \zero_{k/2})$
\end{enumerate}

One can see that we need only $O(\log n)$ of such rules.

Since $b(\rowlin(\Mmn))\geq \delta(\rowlin(\Mmn))$ and $b(\Mmn)\leq g(\Mmn)$ for all $\Mmn$, the analogous comparison can be done with this same example on the measure $b$.
\end{example}

\newpage

\section{Access in 2D strings in \texorpdfstring{$O(g_{rl})$}{O(grl)}  space and \texorpdfstring{$O(h)$}{O(h)} time}\label{sec:access}

In Algorithm~\ref{alg:2D-SLP-direct-access} we show the procedure to access any element in a 2D string $\Mmn$.

\begin{algorithm}[hb]\caption{Direct access for 2D SLPs and 2D RLSLPs in $O(h)$ time.}\label{alg:2D-SLP-direct-access}
\begin{algorithmic}[1]
\Require A 2D (RL)SLP $G$, a variable $A$ of $G$, and a cell position $(i,j) \in [1\dd m_A] \times [1 \dd n_A]$.
\Ensure The character $\gexp(A)[i][j]$ at position $(i,j)$ in $\gexp(A)$.
\Function{access}{$G, A, i, j$}
\If{$A \rightarrow a$}
    \State \Return $a$
\EndIf
\If{$A \rightarrow B \ovrt C$}
    \If{$i \le m_B$}
        \State \Return \Call{access}{$G, B, i, j$}
    \Else
        \State \Return \Call{access}{$G, C, i-m_B, j$}
    \EndIf
\EndIf
\If{$A \rightarrow B \ohrz C$}
    \If{$j \le n_B$}
        \State \Return \Call{access}{$G, B, i, j$}
    \Else
        \State \Return \Call{access}{$G, C, i, j-n_B$}
    \EndIf
\EndIf
\If{$A \rightarrow \ovrt^k B$}
    \State \Return \Call{access}{$G, C, i \bmod m_B, j$}
\EndIf
\If{$A \rightarrow \ohrz^k B$}
    \State \Return \Call{access}{$G, C, i, j \bmod n_B$}
\EndIf
\EndFunction
\end{algorithmic}
\end{algorithm}

\end{document}